\documentclass[10pt,conference]{IEEEtran}
\usepackage{amsfonts}
\usepackage{amssymb,epic,eepic}
\usepackage{amsmath}
\usepackage[usenames,dvipsnames]{pstricks}
\usepackage{epsfig}
\usepackage{pst-grad} 
 \usepackage{pst-plot} %

\newcommand{\ty}{{\mathcal{T}}}
\newcommand{\pp}{{\mathcal{P}}}

\newcommand{\nn}{{\mathbb N}}

\newtheorem{theorem}{Theorem}[section]         
\newtheorem{lemma}[theorem]{Lemma}             
\newtheorem{definition}[theorem]{Definition}   

\begin{document}

\title{Capacity results for compound wiretap channels}
\author{Igor Bjelakovi\'c, Holger Boche, and Jochen Sommerfeld \\[2mm]
\small Lehrstuhl f\"ur theoretische Informationstechnik, Technische Universit\"at M\"unchen, 80290 M\"unchen, Germany\\
Email: \{igor.bjelakovic, boche, jochen.sommerfeld\}@tum.de 
\thanks{}}
\maketitle

\begin{abstract}
We derive a lower bound on the secrecy capacity of the compound wiretap channel with channel state
information at the transmitter which matches the general upper bound on the secrecy capacity of general
compound wiretap channels given by Liang et al. and thus establishing a full coding theorem in this
case. We achieve this with a quite strong secrecy criterion and with a decoder that is robust against the
effect of randomisation in the encoding. This relieves us from the need of decoding the randomisation
parameter which is in general not possible within this model. Moreover we prove a lower bound and a multi-letter converse to 
the secrecy capacity of the compound wiretap channel without channel state information. 
\end{abstract}

\section{Introduction}
Compound wiretap channels are among the simplest non-trivial models incorporating the requirement of
security against a potential eavesdropper while at the same time the legitimate users suffer from
channel uncertainty. They may be considered therefore as a starting point for theoretical investigation
tending towards applications, for example, in wireless systems, a fact explaining an alive research
activity in this area in recent years (cf. \cite{liang}, \cite{bloch} and references therein).\\
In this paper we consider finite families of pairs of channels $\mathfrak{W}=\{(W_t,V_t):t=1,\ldots,
T\}$\footnote{Along the way we will comment what our results look like when applied to widely used class
  of models of the form $\mathfrak{W}=\{(W_t,V_s):t=1,\ldots, T, s=1,\ldots, S \}$ with $T\neq S$ which
  are special case of the model we are dealing with in this paper } with common input alphabet and
possibly different output alphabets. The legitimate users control $W_t$ and the eavesdropper observes the
output of $V_t$. We will be dealing with two communication scenarios. In the first one only the transmitter is 
informed about the index $t$ (channel state information (CSI) at the transmitter) while in the second the
legitimate users have no information about that index at all (no CSI). This setup is a generalisation of
Wyner's \cite{wyner-wire} wiretap channel.\\
Our contributions are summarised as follows: In \cite{liang} a general upper bound on the capacity of compound wiretap channel 
as the minimum secrecy
capacity of the involved wiretap channels was given. We prove
in Section \ref{sec:csi} that the models whose secrecy capacity matches this upper bound contain all compound wiretap
channels with CSI at the transmitter. At the same time we achieve this bound with a substantially
stronger security criterion called strong secrecy which has been
employed already in \cite{csis96}, \cite{maurer-wolf}, \cite{cai-winter-yeung}, and
\cite{devetak}. Indeed, our security proof follows closely that developed in \cite{devetak} for single
wiretap channel with classical input and quantum output. 
In order to achieve secrecy we follow the common approach according to which randomised encoding is a permissible operation. 
The impact of randomisation at the legitimate decoder's site is
usually compensated  by communicating to her/him the outcome of the random experiment performed. However
, in the case of compound wiretap channel with CSI at the transmitter this strategy does not work as is illustrated by an 
example in Section \ref{sec:csi}. We
resolve this difficulty by developing a decoding strategy which is independent of the particular channel realisation and 
is insensitive to randomisation while 
decoding just at the optimal secrecy rate for all channels $\{W_t: t=1,\ldots, T  \}$ simultaneously.\\
Moreover, a slight modification of our proofs allows us to derive a lower bound on the capacity of the compound 
wiretap channel without CSI. Additionally, we give a multi-letter converse to the coding theorem.
This is content of Section \ref{sec:no-csi}. Due to space limitation we cannot provide proofs 
here and refer to the online supporting material \cite{bjela2}.\\
Our results are easily extended to arbitrary sets (even uncountable) of wiretap channels via standard
approximation techniques \cite{blackw}. 


\section{Compound wiretap channel}\label{b}

Let $A,B,C$ be finite sets and $\theta=\{1, \ldots, T\}$ an index set. We consider two families of channels 
$W_t:A\to\mathcal{P}(B)$\footnote{$\mathcal{P}(B)$ denotes the set of probability distributions on $B$}, 
$V_t:A\to \mathcal{P}(C)$, $t\in \theta$, which we collectively abbreviate by $\mathfrak{W}$ and call
the compound wiretap channel generated by the given families of channels. 
Here the first family represents the communication link to the legitimate receiver while the output of
the latter is under control of the eavesdropper. In the rest of the paper expressions like $W_t^{\otimes
  n}$ or $V_t^{\otimes n}$ stand for the $n$-th memoryless extension of the stochastic matrices $W_t$,
$V_t$. \\
An $(n,J_n)$ code for the compound wiretap channel $\mathfrak{W}$ consists of a stochastic encoder 
$E:\mathcal{J}_n\to \mathcal{P}(A^n)$ (a stochastic matrix) with a message set
$\mathcal{J}_n:=\{1,\ldots, J_n \}$ and a collection of mutually 
disjoint decoding sets $\{D_j\subset B^n:j\in\mathcal{J}_n  \}$. The maximum error probability of a
$(n,J_n)$ code $\mathcal{C}_n$ is given by
\begin{equation}\label{eq:error-no-csi}
 e(\mathcal{C}_n):= \max_{t\in \theta}  \, \max_{j \in \mathcal{J}_n} \sum_{x^n\in A^n} E(x^n|j)
W_t^{\otimes n}(D_j^c| x^n).  
\end{equation} 
I.e. neither the sender nor the receiver have CSI.\\
If CSI is available at the transmitter the
probability of error in (\ref{eq:error-no-csi}) changes to
\begin{equation}\label{eq:error-csi}
e_{\textup{CSI}}(\mathcal{C}_n):= \max_{t\in \theta} \, \max_{j \in\mathcal{J}_n} \sum_{x^n\in A^n}
E_t(x^n|j)
W_t^{\otimes n}(D_j^c| x^n).
\end{equation} 
Notice, however, that the decoder appearing in (\ref{eq:error-csi}) has to be universal, i.e. independent of 
the channel index $t$ in accordance to the fact that
no CSI is available at the receiver.\\
We assume throughout the paper that the eavesdropper always knows which channel is in use.
\begin{definition}\label{code}
A non-negative number $R$ is an achievable secrecy rate for $\mathfrak{W}$
with or without CSI respectively if there is a sequence $(\mathcal{C}_n)_{n\in\nn}$ of $(n,J_n)$ codes such that
\[ \lim_{n\to\infty} e_{\textup{CSI}}(\mathcal{C}_n)=0 \textrm{ resp.  } \lim_{n\to\infty} e(\mathcal{C}_n)=0,\]
\begin{equation}\label{eq:strong-secrecy-def} 
\liminf_{n\to\infty}\frac{1}{n}\log J_n\ge R\textrm{ and }\lim_{n\to\infty}\max_{t\in\theta}I(J; Z_t^n)=0, 
\end{equation}
where $J$ is a uniformly distributed random variable taking values in $\mathcal{J}_n$ and $Z_t^{n}$ are the resulting
random variables at the output of eavesdropper's channel $V_t^{\otimes n}$.\\
The secrecy capacity in either scenario is given by the largest achievable secrecy rate and is denoted by
$C_S(\mathfrak{W})$ and $C_{S, CSI}(\mathfrak{W})$.
\end{definition}
\emph{Remark.} A weaker and widely used security criterion is obtained if we replace
(\ref{eq:strong-secrecy-def}) by $\lim_{n\to\infty}\max_{t\in\theta}\frac{1}{n}I(J; Z_t^n)=0 $. We prefer to
follow \cite{csis96}, \cite{cai-winter-yeung},  and \cite{devetak} and require the validity of
(\ref{eq:strong-secrecy-def}). A nice discussion of several secrecy criteria is contained in \cite{bloch}.

\section{Capacity results}

\subsection{Preliminaries}

In what follows we use the notation as well as some properties of \emph{typical} and \emph{conditionally
  typical} sequences from \cite{csis2}. For $p\in\mathcal{P}(A)$, $W:A\to\mathcal{P}(B)$, $x^n\in A^n$,
and $\delta>0$  we denote by $\ty_{p,\delta}^n$ the set of typical sequences and by
$\ty_{W,\delta}^n(x^n) $ the set of conditionally typical sequences given $x^n$.\\
The basic properties of these sets that are needed in the sequel are summarised in the following three lemmata.
\begin{lemma}\label{typical}
Fixing $\delta > 0$, for every $p \in \pp(A)$ and  $W:A \to \pp(B)$ we have
\begin{eqnarray}
p^{\otimes n}(\ty_{p,\delta}^n) & \geq &1- (n+1)^{|A|} 2^{-nc\delta^2} \\
W^{\otimes n}(\ty_{W,\delta}^n(x^n)|x^n) & \geq &1-(n+1)^{|A||B|} 2^{-nc \delta^2}
\end{eqnarray}
for all $x^n\in A^n$ with $c=1/(2\ln 2)$. In particular, there is $n_0\in\nn$ such that for each
$\delta>0$ and $n>n_0$
\begin{eqnarray}
 p^{\otimes n}(\ty_{p,\delta}^n) & \geq &1- 2^{-nc'\delta^2} \\
W^{\otimes n}(\ty_{W,\delta}^n(x^n)|x^n) & \geq &1- 2^{-nc' \delta^2}
\end{eqnarray}
holds with $c'=\frac{c}{2}$.
\end{lemma}
\begin{proof}
 Standard Bernstein-Sanov trick using the properties of types from  \cite{csis2} and Pinsker's inequality. 
The details can be found in \cite{wyrem} and references therein for example.
\end{proof}
Recall that for $p\in\pp(A)$ and $W:A\to\pp(B)$ $pW\in\pp(B)$ denotes the output distribution generated
by $p$ and $W$ and that  $x^n \in \ty^n_{p,\delta}$ and $y^n \in \ty^n_{W,\delta}(x^n)$ imply that $y^n
\in \ty^n_{pW,2|A|\delta}$. 
\begin{lemma}\label{alpha-beta} 
Let $x^n\in \ty^n_{p,\delta}$, then for $V:A\to\pp(C)$
\begin{eqnarray}
|\ty_{pV,2|A|\delta}^n| &\leq& \alpha^{-1}\\
V^n(z^n|x^n) &\leq& \beta \quad \textrm{for all} \quad z^n \in \ty^n_{V,\delta}(x^n)
\end{eqnarray} 
hold where
\begin{eqnarray}
\alpha &=&2^{-n(H(pV)+f_1(\delta))}\label{eq:def-alpha}\\
\beta &=&2^{-n(H(V|p)-f_2(\delta))}\label{eq:def-beta}
\end{eqnarray}
with universal $f_1(\delta),f_2(\delta)>0$ satisfying
$\lim_{\delta\to\infty}f_1(\delta)=0=\lim_{\delta\to\infty}f_2(\delta)$. 
\end{lemma}
\begin{proof}
Cf. \cite{csis2}.
\end{proof}
In addition we need a further lemma which will be used to determine the rates at which reliable
transmission to the legitimate receiver is possible.
\begin{lemma}\label{output-bound}
Let $p, \tilde{p} \in \mathcal{P}(A)$ and two stochastic matrices $W, \widetilde{W}:A \to \mathcal{P}(B)$
be given. Further let $q,\tilde{q} \in \mathcal{P}(B)$ be the output distributions, the former generated by $p$ and
$W$ and the latter by $\tilde{p}$ and $\widetilde{W}$. Fix $\delta \in (0,\frac{1}{4|A||B|})$. Then for
every $n \in \nn$
\begin{equation}
q^{\otimes n}(\ty^n_{\widetilde{W}, \delta}(\tilde{x}^n)) \leq (n+1)^{|A||B|}
  2^{-n(I(\tilde{p},\widetilde{W})-f(\delta))}
\end{equation}
for all $\tilde{x}^n \in \ty^n_{\tilde{p},\delta}$ and 
\begin{equation}
q^{\otimes n}(\ty^n_{W, \delta} (x^n)) \leq (n+1)^{|A||B|} 
  2^{-n(I(p,W)-f(\delta))}  
\end{equation}
for all $x^n \in \ty^n_{p, \delta}$ holds for a universal $f(\delta) >0$ and $\lim_{\delta\to 0}
f(\delta)=0$.
\end{lemma}
\begin{proof}
Cf. \cite{wyrem}.
\end{proof}
The last lemma is a standard result from large deviation theory.
\begin{lemma}\label{chernoff}(Chernoff bounds)
Let $Z_1,\ldots,Z_L$ be i.i.d. random variables with values in $[0,1]$ and expectation
$\mathbb{E}Z_i=\mu$, and $0<\epsilon<\frac{1}{2}$. Then it follows that
\begin{equation}
Pr \left\{ \frac{1}{L} \sum^L_{i=1} Z_i \notin [(1\pm\epsilon)\mu]  \right\} \leq 2\exp \left( -L\cdot
\frac{\epsilon^2\mu}{3} \right), \nonumber
\end{equation}
\end{lemma}
where $[(1\pm\epsilon)\mu]$ denotes the interval $[(1 - \epsilon)\mu, (1 + \epsilon)\mu].$\\

\subsection{CSI at the transmitter}\label{sec:csi}

The main result in this section is the following theorem.
\begin{theorem}\label{CSI-code}
The secrecy capacity of the compound wiretap channel $\mathfrak{W}$ with CSI at the transmitter is given by
\begin{equation}
C_{S,CSI}(\mathfrak{W})= \min_{t \in \theta} \max_{V\rightarrow X \rightarrow (YZ)_t}(I(V,Y_t)-I(V,Z_t)).
\end{equation}
\end{theorem}
Notice first that the inequality
\[ C_{S,CSI}(\mathfrak{W})\le\min_{t \in \theta} \max_{V\rightarrow X \rightarrow (YZ)_t}(I(V,Y_t)-I(V,Z_t)) \]
is trivially true since we cannot exceed the secrecy capacity of the worst wiretap channel in the family
$\mathfrak{W}$. This has been already pointed out in \cite{liang}. The rest of this section is devoted to
the proof of the achievability.
\begin{proof}
We choose $p_1, \ldots, p_T \in \mathcal{P}(A)$ and define new
distributions on $A^n$ by
\begin{equation}\label{eq:pruned(t)}
p'_t(x^n):= \left \{ \begin{array}{ll}
\frac{p^{\otimes n}_t(x^n)}{p^{\otimes n}_t(\ty^n_{p_t,\delta})} & \textrm{if $x^n \in \ty^n_{p_t,\delta}$},\\
0 &  \textrm{otherwise}
\end{array} \right. .
\end{equation} 
Define then for $z^n\in C^n$
\begin{equation}
\tilde{Q}_{t,x^n}(z^n):=V_t^n(z^n|x^n)\cdot\mathbf{1}_{\ty^n_{V_t,\delta}(x^n)}(z^n)
\end{equation}
on $C^n$. Additionally, we set for $z^n\in C^n$
\begin{equation}\label{eq:theta-prime}
\Theta'_t(z^n) =\sum_{x^n \in \ty^n_{p_{t},\delta}} p'_t(x^n)\tilde{Q}_{t,x^n}(z^n).
\end{equation}
Now let $S:=\{z^n\in C^n: \Theta'_t(z^n) \geq \epsilon\alpha_t \}$ where
$\epsilon=2^{-nc'\delta^2}$ (cf. Lemma \ref{typical}) and $\alpha_t$ is from (\ref{eq:def-alpha}) in
Lemma \ref{alpha-beta} computed with respect to $p_t$ and $V_t$. By lemma $3.2$ the support of
$\Theta'_t$ has cardinality $\leq \alpha^{-1}_t$ since for each $x^n\in \ty^n_{p_{t},\delta}$ it holds that $\ty^n_{V_{t},\delta}(x^n)\subset \ty^n_{p_{t}V_{t}, 2|A|\delta}$, 
which implies that $\sum_{z^n \in S}\Theta_t(z^n) \geq
1-2\epsilon$, if 
\begin{eqnarray}
\Theta_t(z^n)&=&\Theta'_t(z^n)\cdot\mathbf{1}_S(z^n) \quad \textrm{and} \nonumber\\
Q_{t,x^n}(z^n)&=&\tilde{Q}_{t,x^n}(z^n) \cdot \mathbf{1}_S(z^n).\label{eq:theta}
\end{eqnarray}
Now for each $t \in \theta$ define $ J_n  \cdot L_{n,t} $ i.i.d. random
variables $X^{(t)}_{jl}$ with $j \in [J_n]:= \{1,\ldots, J_n   \}$ and $l \in [L_{n,t}]:=\{1,\ldots,
L_{n,t}  \}$ each of them distributed according to $p'_t $ with 
\begin{eqnarray}
J_n&=& \left\lfloor2^{n[\min_{t \in \theta}(I(p_t,W_t)-I(p_t,V_t))-\tau]}\right\rfloor\label{eq:J-n} \\
L_{n,t}&=&\left\lfloor 2^{n[I(p_t,V_t)+\frac{\tau}{4}]}\right\rfloor\label{eq:L-n-t}
\end{eqnarray} 
for $\tau>0$. Moreover we suppose that the random matrices $\{X^{(t)}_{j,l}\}_{j\in[J_n],l\in[L_{n,l}]}  $ and
$ \{X^{(t')}_{j,l}\}_{j\in[J_n],l\in[L_{n,l}]}   $ are independent for $t\neq t'$.
Now it is obvious from (\ref{eq:theta-prime}) and the definition of the set $S$ that for any $z^n\in S$
$\Theta_t(z^n)=\mathbb{E}_t Q_{t,X^{(t)}_{jl}}(z^n)\ge \epsilon \alpha_t$ if $\mathbb{E}_t$ is the expectation
value with respect to the distribution $p'_t$. For the random variables $\beta^{-1}_t
Q_{t,X^{(t)}_{jl}}(z^n)$ ($\beta_t$ from \eqref{eq:def-beta}) define the event
\begin{equation}\label{eq:first-iota}
\iota_j(t)=\bigcap_{z^n \in C^n} \left\{\frac{1}{L_{n,t}}\sum_{l=1}^{L_{n,t}} Q_{t,X^{(t)}_{jl}}(z^n) \in [(1 \pm
  \epsilon) \Theta_t(z^n)]\right \},
\end{equation}
and keeping in mind that  $\Theta_t(z^n) \geq \epsilon \alpha_t$ for all $z^n \in S$. It follows that
for all $j \in [J_n]$ and for all $t \in \theta$
\begin{equation}\label{eq:dev-t}
\textrm{Pr}\{ (\iota_j(t))^c\} \leq 2 |C|^n
\exp  \Big(- L_{n,t} \frac{ 2^{-n[I(p_t,V_t)+g(\delta)]}}{3}   \Big)
\end{equation}
by Lemma \ref{chernoff}, Lemma \ref{alpha-beta}, and our choice $\epsilon=2^{-nc'\delta^2}$ with 
$g(\delta):=f_1(\delta)+f_2(\delta)+3c'\delta^2$. Making $\delta>0$ sufficiently small we have for all 
sufficiently large $n\in \nn$ and all $t\in \theta$ $L_{n,t} 2^{-n[I(p_t,V_t)+g(\delta)]}\ge 2^{n\frac{\tau}{8}}$. 
Thus, for this choice of $\delta$ the RHS of (\ref{eq:dev-t}) is double exponential in $n$ uniformly in
$t\in\theta$   and can be made smaller than $\epsilon J_n^{-1}$ for all $j \in [J_n]$ and all
sufficiently large $n \in \nn$. I.e. 
\begin{equation}\label{eq:summand-1}
 \textrm{Pr}\{ (\iota_j(t))^c\} \leq \epsilon J_n^{-1} \quad \forall t\in\theta.
\end{equation} 
Let us turn now to the coding part of the problem.
Let $p'_t\in\mathcal{P}(A^n)$ be given as in \eqref{eq:pruned(t)}. We abbreviate
 $\mathcal{X}:=\{X^{(t)}\}_{t\in \theta}$ for the family of random
matrices $X^{(t)}=\{ X_{jl}^{(t)} \}_{j\in [J_n], l\in [L_{n,t}]}$ whose components are i.i.d. according to
$p'_t$. We will show now how the reliable transmission of the message $j\in [J_n]$ can be achieved when
randomising over the index $l\in L_{n,t}$ without any attempt to decode the latter at the legitimate receiver's site. To this
end let us define for each $j\in[J_n]$ a random set
\begin{equation}\label{eq:achiev-is-4}
 D'_j(\mathcal{X}):=\bigcup_{s\in\theta}\bigcup_{k\in[L_{n,s}]} \ty_{W_s,\delta}^{n}(X_{jk}^{(s)}),
\end{equation} 
and the subordinate random decoder $\{ D_j(\mathcal{X}) \}_{j\in[J_n]} \subseteq B^n$ is given by
\begin{equation}\label{eq:achiev-is-5}
 D_j(\mathcal{X}):=D'_j(\mathcal{X})\cap \Big( \bigcup_{\substack{j'\in [J_n] \\ j'\neq j}}
   D'_{j'}(\mathcal{X}) \Big)^c. 
\end{equation} 
Consequently we can define the random average probabilities of error for a specific channel $t\in \theta$ by
\begin{equation}\label{eq:error-t}
\lambda_n^{(t)}(\mathcal{X}):=\frac{1}{J_n}\sum_{j\in [J_n]}\frac{1}{L_{n,t}}\sum_{l\in[L_{n,t}]} 
  W_t^{\otimes n}((D_j(\mathcal{X}) )^{c}| X_{jl}^{(t)} ) .
\end{equation}
 Now (\ref{eq:achiev-is-5}) implies for each $t\in \theta$ and $l\in [L_{n,t}]$
\begin{equation}\label{eq:achiev-is-6}
\begin{split}
 W_t^{\otimes n}&((D_j(\mathcal{X}) )^{c}| X_{jl}^{(t)} )  \\
 &\le  W_t^{\otimes n} (\bigcap_{s\in \theta } \bigcap_{k\in[L_{n,s}]}  
(\ty_{W_s,\delta}^{n}(X_{jk}^{(s)}))^c| X_{jl}^{(t)}        ) \\
&+\sum_{\substack{ j'\in[J_n]\\ j'\neq j }}\sum_{s\in \theta}\sum_{k\in[L_{n,s}]}W_t^{\otimes n}(
\ty_{W_s,\delta}^n  (X_{j'k}^{(s)})|X_{jl}^{(t)}     )  \\
&\le W_t^{\otimes n}( (\ty_{W_t,\delta}^{\otimes n}( X_{jl}^{(t)} ))^c| X_{jl}^{(t)}  ) \\
& +\sum_{\substack{ j'\in[J_n]\\ j'\neq j }}\sum_{s\in \theta}\sum_{k\in[L_{n,s}]}W_t^{\otimes n}(
\ty_{W_s,\delta}^n  (X_{j'k}^{(s)})|X_{jl}^{(t)}     ).
\end{split}
\end{equation}
By Lemma \ref{typical} and the independence of all involved
random variables we obtain
\begin{equation}\label{eq:achiev-is-7}
\begin{split}
 &\mathbb{E}_{\mathcal{X}}  (  W_t^{\otimes n}((D_j(\mathcal{X}) )^{c}| X_{jl}^{(t)} )) \\
 & \leq (n+1)^{|A||B|}\cdot 2^{-nc\delta^2} \\
 & + \sum_{\substack{ j'\in[J_n]\\ j'\neq j }}\sum_{s\in \theta}\sum_{k\in[L_{n,s}]} \mathbb{E}_{X_{j'k}^{(s)}}
      \mathbb{E}_{X_{jl}^{(t)}}  W_t^{\otimes n}( \ty_{W_s,\delta}^n  (X_{j'k}^{(s)})|X_{jl}^{(t)}     ).
\end{split}
\end{equation} 
We shall find now for $j'\neq j$ an upper bound on
\begin{equation}\label{eq:achiev-is-8}
\begin{split}
\mathbb{E}_{X_{jl}^{(t)}} &  W_t^{\otimes n}( \ty_{W_s,\delta}^n 
(X_{j'k}^{(s)})|X_{jl}^{(t)}     ) \\
 &= \sum_{x^n\in A^n} p'_t(x^n)  W_t^{\otimes n}( \ty_{W_s,\delta}^n (X_{j'k}^{(s)})|x^n  ) \\
&\le \sum_{x^n\in A^n}\frac{p_t^{\otimes n} (x^n)}{p_t^{\otimes n}(\ty_{p_t,\delta}^n)}
W_t^{\otimes n}( \ty_{W_s,\delta}^n (X_{j'k}^{(s)})|x^n  )\\
&=  \frac{q_t^{\otimes n}(\ty_{W_s,\delta}^n 
(X_{j'k}^{(s)})  )  }{ p_t^{\otimes n}(\ty_{p_t,\delta}^n)}.
\end{split}
\end{equation}
By Lemma \ref{typical} and by Lemma \ref{output-bound} for any $t,s \in \theta$ we have
\begin{equation}\label{eq:achiev-is-9}
\begin{split}
p_t^{\otimes n}(\ty_{p_t,\delta}^n) & \ge 1-(n+1)^{|A|}\cdot 2^{-nc\delta^2} \\
q_t^{\otimes n}(\ty_{W_s,\delta}^n 
(X_{j'k}^{(s)})  ) & \le (n+1)^{|A||B|}\cdot 2^{-n (I(p_s,W_s) -f(\delta))} 
\end{split}
\end{equation}
with a universal $f(\delta)>0$ satisfying $\lim_{\delta\to 0}f(\delta)=0$ since $X_{j',k}^{(s)}\in 
\ty_{p_s,\delta}^n $ with probability 1. Thus inserting this into (\ref{eq:achiev-is-8}) we obtain
\begin{equation}\label{eq:achiev-is-10}
\begin{split}
 \mathbb{E}_{X_{jl}^{(t)}} & W_t^{\otimes n}( \ty_{W_s,\delta}^n 
(X_{j'k}^{(s)})|X_{jl}^{(t)}     )\\
 &\le \frac{(n+1)^{|A||B|}}{1-(n+1)^{|A|}\cdot 2^{-nc\delta^2}}\cdot 2^{-n (I(p_s,W_s) -f(\delta))}
\end{split}
\end{equation}
for all $s,t\in \theta$, all $j'\neq j$, and all $l\in [L_{n,t}], k\in [L_{n,s}]$. Now by defining
$\nu_n(\delta):=(n+1)^{|A||B|}\cdot 2^{-nc\delta^2}$ and $\mu_n(\delta):=1-(n+1)^{|A|}\cdot
2^{-nc\delta^2}$ thus  for each $t\in \theta$, $l\in [L_{n,t}]$, and  $j\in[J_n]$ (\ref{eq:achiev-is-7}) and
(\ref{eq:achiev-is-8}) lead to
{\small
\begin{equation}\label{eq:achiev-is-11}
\begin{split}
&\mathbb{E}_{\mathcal{X}}  (  W_t^{\otimes n}((D_j(\mathcal{X}) )^{c}| X_{jl}^{(t)} )) \\
& \le \nu_n(\delta)
+ \frac{(n+1)^{|A||B|}}{\mu_n(\delta)}J_n \sum_{s\in \theta} L_{n,s}  2^{-n (I(p_s,W_s) -f(\delta))} \\
&\le \nu_n(\delta)  \\
&+ \frac{(n+1)^{|A||B|}}{\mu_n(\delta)} T\cdot  J_n \cdot 
2^{-n (\min_{s\in\theta}( I(p_s,W_s)-I(p_s,V_s) ) -f(\delta)-\frac{\tau}{4}   )} \\
&\le \nu_n(\delta) 
 + \frac{(n+1)^{|A||B|}}{\mu_n(\delta)} T \cdot 2^{-n(\tau -f(\delta)-\frac{\tau}{4} )} \\
&\le \nu_n(\delta)
 + \frac{(n+1)^{|A||B|}}{\mu_n(\delta)} T \cdot 2^{-n\frac{\tau}{2}} 
\end{split}
\end{equation}
}where we have used (\ref{eq:L-n-t}), (\ref{eq:J-n}), and we have chosen $\delta>0$ small enough to
ensure that $\tau -f(\delta)-\frac{\tau}{4}\ge \frac{\tau}{2} $. Defining $a=a(\delta,\tau):=\frac{\min\{ c\delta^2,
  \frac{\tau}{4}  \}}{2}$ we can find $n(\delta,\tau,|A|,|B|)\in\nn$ such that  for all $n\ge n(\delta,\tau,|A|,|B|) $ 
\begin{equation}\label{eq:achiev-is-12}
 \mathbb{E}_{\mathcal{X}}(  W_t^{\otimes n}((D_j(\mathcal{X}) )^{c}| X_{jl}^{(t)} ))\le T \cdot 2^{-n a}
\end{equation} 
holds for all $t\in \theta$, $l\in [L_{n,t}]$, and $j\in [J_n]$. Consequently, for any $t\in \theta$ we obtain
\begin{equation}\label{eq:error}
\mathbb{E}_{\mathcal{X}}  (\lambda^{(t)}_n(\mathcal{X})) \le T\cdot 2^{-na}.
\end{equation} 
We define for any $t\in\theta$ an event $\iota_0 (t)=\{\lambda_n^{(t)}(\mathcal{X}) \leq \sqrt{T}2^{-n\frac{a}{2}}\}$.
Then using the Markov inequality and (\ref{eq:error}) yields
\begin{equation}\label{eq:summand-2}
\textrm{Pr}\{ (\iota_0(t))^c\} \leq \sqrt{T}2^{-n\frac{a}{2}}.
\end{equation}
Set $\iota:=\bigcap_{t\in\theta}\bigcap_{j=0}^{J_n}\iota_j(t)$. Then with (\ref{eq:summand-1}), (\ref{eq:summand-2}), 
and applying the union bound we obtain 
\begin{equation*}
\textrm{Pr}\{ \iota^c\} 
\leq \sum_{t\in\theta}\sum_{j=0}^{J_n}\textrm{Pr}\{(\iota_j(t))^c\} \leq
 T  \epsilon +T^{\frac{3}{2}} 2^{-n \frac{a}{2}} 
\le   T^2 \cdot 2^{-nc''}
\end{equation*}
for a suitable positive constant $c''>0$ and all sufficiently large $n\in\nn$.\\
Hence, we have shown that for each $t\in \theta$ there exist realisations 
$\{ (x^{(t)}_{jl})_{j\in[J_n], l\in [L_{n,t}]}: t\in\theta\}\in \iota$ of 
$\mathcal{X}$.\\
We show that the secrecy level is fulfilled uniformly in $t\in \theta$  for any particular 
$\{ (x^{(t)}_{jl})_{j\in[J_n], l\in [L_{n,t}]}: t\in\theta\}\in \iota$ Denoting by $||\cdot ||$ the variational
 distance we have
\begin{multline}\label{eq:secure}
\Big\| \frac{1}{L_{n,t}}\sum^{L_{n,t}}_{l=1} V^n_t(\cdot|x^{(t)}_{jl})-\Theta_t(\cdot)\Big\|\\
\leq \frac{1}{L_{n,t}}\sum^{L_{n,t}}_{l=1} \left\| V^n_t(\cdot|x^{(t)}_{jl})-\tilde{Q}_{t,x^{(t)}_{jl}}(\cdot)\right\|\\
+ \Big\| \frac{1}{L_{n,t}}\sum^{L_{n,t}}_{l=1} \big(\tilde{Q}_{t,x^{(t)}_{jl}}(\cdot)- Q_{t,x^{(t)}_{jl}}(\cdot)\big)
\Big\| \\
+ \Big\| \frac{1}{L_{n,t}}\sum^{L_{n,t}}_{l=1} Q_{t,x^{(t)}_{jl}}(\cdot)-\Theta_t(\cdot)\Big\| \leq 5\epsilon,
\end{multline}
In the first term the functions $V^n_t(\cdot |x_{jl}^{(t)})$ and $\tilde{Q}_{t,x_{jl}^{(t)} }(\cdot)$ differ if $z^n \notin
\ty^n_{V_t,\delta}(x_{jl}^{(t)})$, so it makes a contribution of $\epsilon$ to the bound.
In the second term $\tilde{Q}_t$ and $Q_t$ are different for $z^n \notin S$ and because $\iota_j(t)$ and
 $ \sum_{z^n \in S}\Theta_t(z^n) \geq
1-2\epsilon $ imply
$\frac{1}{L_{n,t}}\sum^{L_{n,t}}_{l=1} \sum_{z^n \in S} Q_{t,x^{(t)}_{jl}}(z^n) \geq 1-3\epsilon$, the
second term is bounded by $3\epsilon$. The third term is bounded by $\epsilon$ which follows directly
from \eqref{eq:first-iota}.\\
For any $\{ (x^{(t)}_{jl})_{j\in[J_n], l\in [L_{n,t}]}: t\in\theta\}\in \iota$ with the corresponding decoding sets
$\{ D_j: j\in [J_n] \}$ it follows by construction that
\begin{equation}\label{eq:expurg-1}
\frac{1}{J_n}\sum_{j\in[J_n]}\frac{1}{L_{n,t}}\sum_{l\in [L_{n,t}]}W^{\otimes n}_t(D_j^c|
 x^{(t)}_{jl})\le \sqrt{T}
 \cdot 2^{-n a'}
\end{equation} 
is fulfilled for all $t \in \theta$ with $a'>0$, which means that we have found a $(n,J_n)$ code with
average error probability tending to zero for $n \in \nn$ sufficiently large for any channel
realisation. Now by a standard expurgation scheme we show that this still holds for the maximum error
probability. We define the set
\begin{equation}\label{eq:expurg-3}
G_t:=\{j \in J_n: \frac{1}{L_{n,t}}\sum_{l\in [L_{n,t}]}W^{\otimes n}_t(D_j^c|x^{(t)}_{jl}) \leq \sqrt{\eta}  \}
\end{equation} 
with $\eta:=\sqrt{T}\cdot 2^{-na'}$ and denote its complement as $B_t:=G_t^c$ and the union of all
complements as $B=\bigcup_{t \in \theta} B_t$.
Then \eqref{eq:expurg-1} and \eqref{eq:expurg-3} imply that
\begin{equation}
\eta \geq \frac{1}{J_n}\sum_{j\in[J_n]}  \frac{1}{L_{n,t}}\sum_{l\in [L_{n,t}]}W^{\otimes n}_t(D_j^c|
 x^{(t)}_{jl}) \geq \frac{|B_t|}{J_n} \sqrt{\eta}
\end{equation}
for all $t \in \theta$ and by the union bound it follows that $|B|\leq T\cdot \sqrt{\eta}\cdot J_n$ .
After removing all $j\in B$ (which are at most a fraction of $T^{\frac{5}{4}}
2^{-n\frac{a'}{2}}$ of $J_n$) and relabeling we obtain a new $(n,\tilde{J}_n)$ code $(E_j.D_j)_{j\in [\tilde{J}_n]}$ 
without changing the rate. The maximum error probability of the new code fulfills for
sufficiently large $n \in \nn$
\begin{equation}
\max_{t \in \theta} \, \max_{j \in [\tilde{J_n}]} \frac{1}{L_{n,t}} \sum_{l \in [L_{n,t}]} W^{\otimes n}_t(D^c_j |
x^{(t)}_{jl}) \leq T^{\frac{1}{4}} \cdot 2^{-n\frac{a'}{2}}.
\end{equation}
On the other hand, if we set $\hat{V}^n_t(z^n|(j,l)):= V^n_t(z^n|x^{(t)}_{jl})$ and further define $\hat{V}^n_{t,j}(z^n) = 
\frac{1}{L_{n,t}}\sum ^{L_{n,t}}_{l=1}\hat{V}^n_t(z^n|(j,l))$,  $
\bar{V}^n_t(z^n) =\frac{1}{\tilde{J}_n} \sum^{\tilde{J}_n}_{j=1}\hat{V}^n_{t,j}(z^n)$,
we obtain that
\begin{equation}
\|\hat{V}^n_{t,j}-\bar{V}^n_t\| \leq  \|  \hat{V}^n_{t,j}-\Theta_t\| + \| \Theta_t- \bar{V}^n_t\|
\le 10\epsilon,
\end{equation}
for all  $j\in [\tilde{J}]_n, t\in \theta$ with $\epsilon=2^{-nc'\delta^2}$ where we have used the convexity
of the variational distance and \eqref{eq:secure} which still applies by our expurgation procedure. For a uniformly distributed 
 random variable 
$J \in [\tilde{J}_n]$, we obtain with Lemma $2.7$ of \cite{csis2} (uniform continuity of the entropy function)
\begin{eqnarray}
I(J;Z^n_t) &=& \sum^{\tilde{J}_n}_{j=1}\frac{1}{\tilde{J}_n}(H(\bar{V}_t^n)-H(\hat{V}_{t,j}^n))\\
                &\leq&-10\epsilon\log(10\epsilon)+10 n\epsilon\log|C|
\end{eqnarray}
uniformly in $t\in \theta$ (for $10\epsilon\leq e^{-1}$). 
Hence the strong secrecy level of the
definition \ref{code} holds uniformly in $t\in\theta$.
Using standard arguments (cf. \cite{csis2} page 409) we then have shown the achievability of
the secrecy rate $\min_{t \in \theta} \max_{V\rightarrow X \rightarrow (YZ)_t}(I(V,Y_t)-I(V,Z_t))$.
\end{proof}
Now with a simple example we will show that indeed identifying both the message  and the randomizing
indices at the legitimate receiver is not possible for all pairs $j\in [J_n]$ and $l \in [L_{n,t}]$. This
is in contrast to the case where we have only one channel to both the legitimate receiver and the
eavesdropper (cf. \cite{devetak}, \cite{csis96}). 
Therefore let $\eta\in [0,1]$ and set $ D_{\eta}:=   
\begin{pmatrix}
1-\eta & \eta\\
\eta & 1-\eta 
\end{pmatrix}
.$
Further, define the channels to the legitimate receiver and to the eavesdropper by
$W_0=D_\eta, \eta \in [0,\frac{1}{2}),$  $ V_0:=D_\tau W_0,\tau \in [0,\frac{1}{2}), \tau \approx 0$,
and $W_1:= D_{\hat{\tau}} V_0=D_{\hat{\tau}}D_\tau W_0, V_1:= \begin{pmatrix}
\frac{1}{2} & \frac{1}{2}\\
\frac{1}{2} & \frac{1}{2}
\end{pmatrix}$. Let $\mathfrak{W}:=\{(W_t,V_t):t=0,1 \}$.
Hence $V_0$ and $W_1$ are degraded versions of $W_0$ and
$I(p,V_1)=0$, for all $ p \in \mathcal{P}(A)$. Now for every $p \in \mathcal{P}(A)$ we can choose $\tau$ small 
enough, such that
$I(p,W_0)-I(p,V_0) < I(p,W_1)$. With $p_0=\big(\frac{1}{2}, \frac{1}{2} \big)$, $\nu>0$ we have in the case of 
CSI at the transmitter
by the defining equations \eqref{eq:J-n} and \eqref{eq:L-n-t}  $J_n= 2^{n[I(p_0,W_0)-I(p_0,V_0))-\nu]}$ and  
$L_{n,0}= 2^{n[I(p_0,V_0)+\frac{\nu}{4}]}$.
But choosing $\hat{\tau}\in (0,1/2]$ large enough we obtain 
$1/n\log (J_nL_{n,0}) =I(p_0,W_0)-  \frac{3\nu}{4} > I(p_0,W_1)=C_{CSI}\{W_0,W_1\}$,
for all small $\nu>0$. Here, $C_{CSI} \{W_0,W_1\}$ stands for the capacity of the compound channel
without secrecy, with CSI at the 
transmitter.\\
\emph{Remark.} Note that for $\mathfrak{W}:=\{ W_t,V_s: t=1,\ldots T, s=1,\ldots S\}$ with
$S\neq T$ and the pair $(s,t)$ known to the transmitter prior to transmission nothing new happens. A
slight modification of the arguments presented above shows that
\[C_{C,CSI}(\mathfrak{W})= \min_{(t,s)} \max_{V\rightarrow X \rightarrow (Y_t Z_s)}(I(V,Y_t)-I(V,Z_s)). \]

\subsection{No CSI}\label{sec:no-csi}
In this final section we describe briefly the capacity results in the case where neither the transmitter
nor the receiver have access to the channel state. We omit the proofs which can be found in the online
supporting material \cite{bjela2}. Achievability part, for example, is shown by a simple modification of
the proof presented in Section \ref{sec:csi}. 
Nevertheless it should be mentioned that the traditional method of proof, i.e. sending messages and the
randomisation parameters to the legitimate receiver, works as well in the present model.
\begin{theorem}\label{no-csi-multi-letter-capacity}
 For any compound wiretap channel $\mathfrak{W}$ we have
\begin{equation*}
  C_S(\mathfrak{W})\ge \max_{U\to X\to Y_t Z_t}( \min_{t \in \theta}I(U,Y_t)-\max_{t \in \theta}I(U,Z_t)),
\end{equation*} 
where $X$ and $Y_tZ_t$  are connected via $(W_t,V_t)$. Consequently, by simple blocking argument we obtain
\begin{equation*}
 C_S(\mathfrak{W})\ge \lim_{n\to\infty} \frac{1}{n}\max_{U\to X\to Y_t^n Z_t^n}( \min_{t \in \theta}I(U,Y_t^n)-
\max_{t \in \theta}I(U,Z_t^n)).
\end{equation*} 
\newpage
Moreover,
\begin{equation*}
 C_S(\mathfrak{W})\le \lim_{n\to\infty} \frac{1}{n}\max_{U\to X\to Y_t^n Z_t^n}( \min_{t \in \theta}I(U,Y_t^n)-
\max_{t \in \theta}I(U,Z_t^n)).
\end{equation*} 
\end{theorem}
Let us consider now the case $\mathfrak{W}:=\{ W_t,V_s: t=1,\ldots T, s=1,\ldots S\}$ with $S\neq T$ and
the pair $(s,t)$ unknown to both the transmitter and the legitimate receiver. If we additionally assume that
each $V_s$ is a degraded version of every $W_t$ then the upper bound in Theorem 2 of \cite{ahlswcsis1}
together with the concavity of $I(X;Y_t|Z_s)$ in the input distribution, which holds in the degraded case, imply that
\begin{equation*}
C_S(\mathfrak{W}) = \max_{p\in\mathcal{P}(A)}(\min_{t}I(p,W_t)-\max_{s}I(p,V_s)).
\end{equation*}
a result which was obtained in \cite{liang} with a weaker notion of secrecy.

\section*{Acknowledgment}
Support by the Deutsche Forschungsgemeinschaft (DFG) via projects 
BO 1734/16-1, BO 1734/20-1, and by the Bundesministerium f\"ur Bildung und Forschung (BMBF) via grant
01BQ1050 is gratefully acknowledged.

\bibliographystyle{IEEEtran}
\bibliography{references}

\end{document}